\documentclass[runningheads]{llncs}

\usepackage[T1]{fontenc}
\usepackage{graphicx}%
\usepackage{hyperref}
\usepackage{amsmath,amssymb,amsfonts}%
\usepackage{xcolor}%
\usepackage{algorithm}%
\usepackage{algorithmicx}%
\usepackage{algpseudocode}%
\usepackage{braket}         
\usepackage{subfig}
\usepackage[qm]{qcircuit}
\usepackage{todonotes}
\usepackage{cite}

\newcommand{\mat}[1]{\left(\begin{matrix}#1\end{matrix}\right)}

\begin{document}
\title{A quantum walk-based scheme for distributed searching on arbitrary graphs}
%
\author{Mathieu Roget\inst{1,2}\orcidID{0000-0003-3890-1552} \and
	Giuseppe Di Molfetta\inst{1,2}\orcidID{0000-0002-1261-7476}}
\institute{Laboratoire d'Informatique et Systèmes, Marseille, France \and
	Université Aix-Merseille, Marseille, France}

\maketitle

\begin{abstract}
A discrete time quantum walk is known to be the single-particle sector of a quantum cellular automaton. Searching in this mathematical framework has interested the community since a long time. However, most results consider spatial search on regular graphs. This work introduces a new quantum walk-based searching scheme, designed to search nodes or edges on arbitrary graphs. As byproduct, such new model allows to generalise quantum cellular automata, usually defined on regular grids, to quantum anonymous networks, allowing a new physics-like mathematical environment for distributed quantum computing.

\end{abstract}

\section*{Introduction}\label{sec:intro}

Quantum Walks~(QW), from a mathematical point of view, coincide with the single-particle sector of quantum cellular automata (QCA), namely a spatially distributed network of local quantum gates. Usually defined on a $d$-dimensional grid of cell, they are known to be capable of universal computation~\cite{arrighi2019overview}. Both quantum walks and quantum cellular automata, with their beginnings in digital simulations of fundamental physics~\cite{di2013quantum, bisio2016quantum}, come into their own in algorithmic search and optimization applications~\cite{santha2008quantum}. Searching using QW has been extensively studied in the past decades, with a wide range of applications, including optimisation~\cite{slate2021quantum} and machine learning~\cite{melnikov2019predicting}. On the other hand, search algorithms as a field independent of quantum walks have recently been used, as subroutines, to solve distributed computational tasks~\cite{gall2018quantum, le2018sublinear, izumi2019quantum}. However, in all these examples, some global information is supposed to be known by each node of the network, such as the size, and usually the network is not anonymous, namely every node has a unique label. Note that, the absence of anonymity and the quantum properties of the algorithm successfully solves the incalculability of certain problems such as the leader election problem. However, encoding global information within a quantum state is generally problematic. In order to address such issue, here we first introduce a new QW-based scheme on arbitrary graphs for rephrasing search algorithms. Then we move to the multi-states generalisation, which successfully implement a QCA-based distributed anonymous protocol for searching problems, avoiding any use of global information.

\paragraph{Contribution}
Section \ref{sec:maths} of this work introduces a Quantum Walk model well suited to search indifferently a node or an edge in arbitrary graphs. In this model the walker's amplitudes are defined onto the graph's edges and ensures 2-dimensional coin everywhere. We detail how this Quantum Walk can be used to search a node or an edge and we show examples of this Quantum Walk on several graphs. In this first part, the Quantum Walk is introduced formally as a discrete dynamical system. In Section \ref{sec:impl} we move to the multi-particle sector, allowing a many quantum states-dynamics over the network, based on the previous model, and leading to a distributed searchign protocol~\ref{sec:impl}. For the coin and the scattering, local quantum circuits and distributed algorithms are provided. The implementation proposed conserves the graph locality and does not require a node or an edge to have global information like the graph size. The nodes (and edges) do not have unique label, and no leader is needed.


\section{Mathematical model of quantum walk on graphs}\label{sec:maths}
This section introduces the quantum walk searching algorithm on undirected connected graphs. First we introduce the quantum walk model on graph and give an example. Then we explain how such a quantum walk can be used to search a marked edge or a marked node. Finally, we show two examples of searching~: in the first one we look for an edge and in the second one we look for a node. Finally we derive the asymptotic complexities.

\subsection{Quantum walk model on graphs}
We consider an undirected connected graph $G = (V,E)$, where $V$ is the set of vertices and $E$ the set of edges. We define the walker's position on the graph's edges and a coin register of dimension two (either $+$ or $-$). The full state of the walker at step $t$ is noted~:
$$
\ket{\Phi_t} = \sum_{(u,v)\in E} \psi_{u,v}^+\ket{(u,v)}\ket{+} + \psi_{u,v}^-\ket{(u,v)}\ket{-}.
$$

The graph is undirected so we indifferently use $\ket{(u,v)}$ and $\ket{(v,u)}$ to name the same state. Similarly for the complex amplitudes $ \psi_{u,v}^+(t) =  \psi_{u,v}^-(t)$. We also introduce a polarity for every edge of $G$. For each of them, we have a polarity function $\sigma$, such that~:
$
\forall (u,v) \in E, \; \sigma(u,v) \in \{+,-\} \text{ and } \sigma(u,v) \neq \sigma(v,u).
$
Figure \ref{fig:tikzedge} illustrates how amplitudes and polarity are placed with edge $(u,v)$ of polarity $\sigma(u,v) = 1$. 
 \begin{figure}[h]
	\centering
	\subfloat[$\sigma(u,v) = +$]{\includegraphics[width=0.4\linewidth]{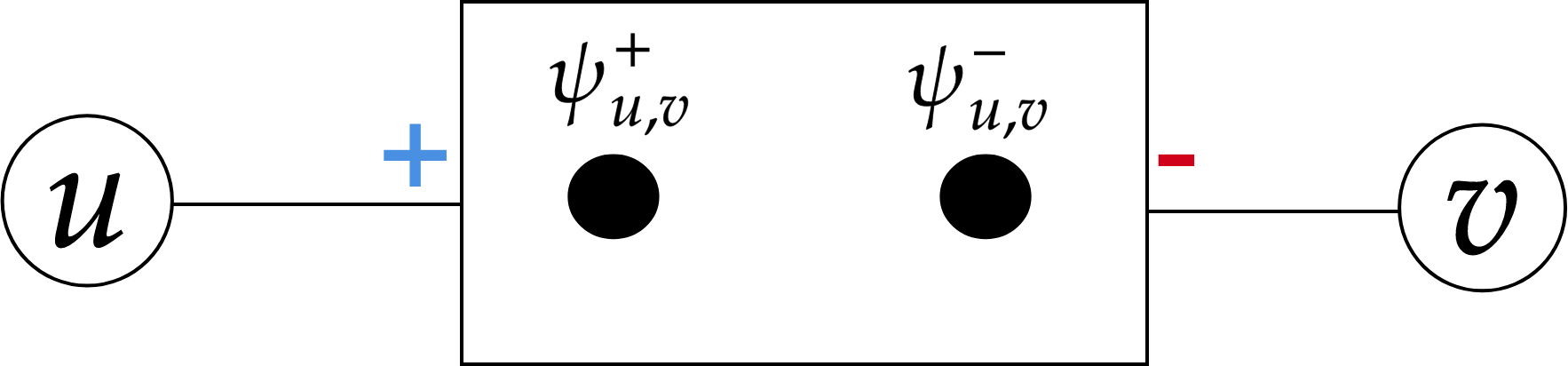}}\qquad
	\subfloat[$\sigma(u,v) = -$]{\includegraphics[width=0.4\linewidth]{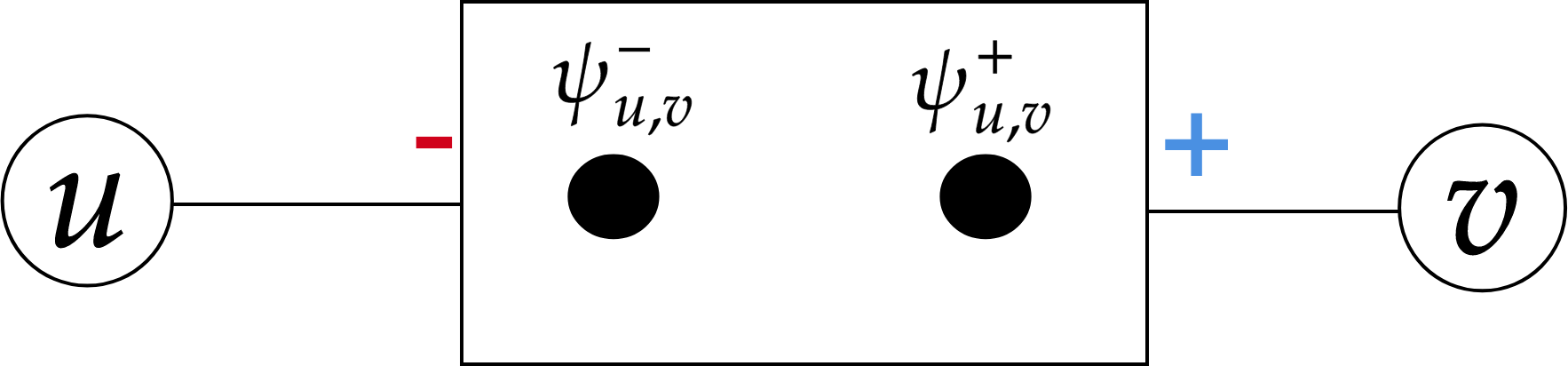}}\\
	\caption{The edge $(u,v)$ and how are placed the amplitudes for the two possible polarities.}
	\label{fig:tikzedge}
\end{figure}
The polarity is a necessary and arbitrary choice made at the algorithm's initialization \textit{independently for every edge}. We discuss why it is necessary and a way to make that choice at the end of this section. 
The full unitary evolution of the walk reads~: $\ket{\Phi_{t+1}} = S \times (I\otimes C) \times \ket{\Phi_{t}}$, where $C$ is the local coin operation acting on the coin register~: 
$$
\forall (u,v)\in E, \; \ket{(u,v)}\ket{\pm} \overset{\text{coin}}{\longmapsto} (I\otimes C) \times \ket{(u,v)}\ket{\pm} =  \ket{(u,v)} (C\ket{\pm}).
$$
and $S$ is the scattering which moves the complex amplitudes $\alpha_{u,v}^\pm$ according to the coin state~:
$$
\forall u \in V,\; \left(\psi_{u,v}^{\sigma(u,v)}\right)_{v\in V} \overset{\text{scattering}}{\longmapsto} D_{\text{deg}(u)}\times \left(\psi_{u,v}^{\sigma(u,v)}\right)_{v\in V},
$$
where $D_n = \left(\frac{2}{n}\right) _{i,j} - I_n$.

As an example of the above dynamics, one can consider the path of size 3 with the nodes $\{u,v,w\}$, with polarity $\sigma(u,v) = \sigma(v,w) = +$. Figure \ref{fig:ring} shows the unitary evolution of the walker from step $t$ to step $t+1$, when the coin coincide the first Pauli matrix $X$.

\begin{figure}
	\centering
	\includegraphics[width=\textwidth]{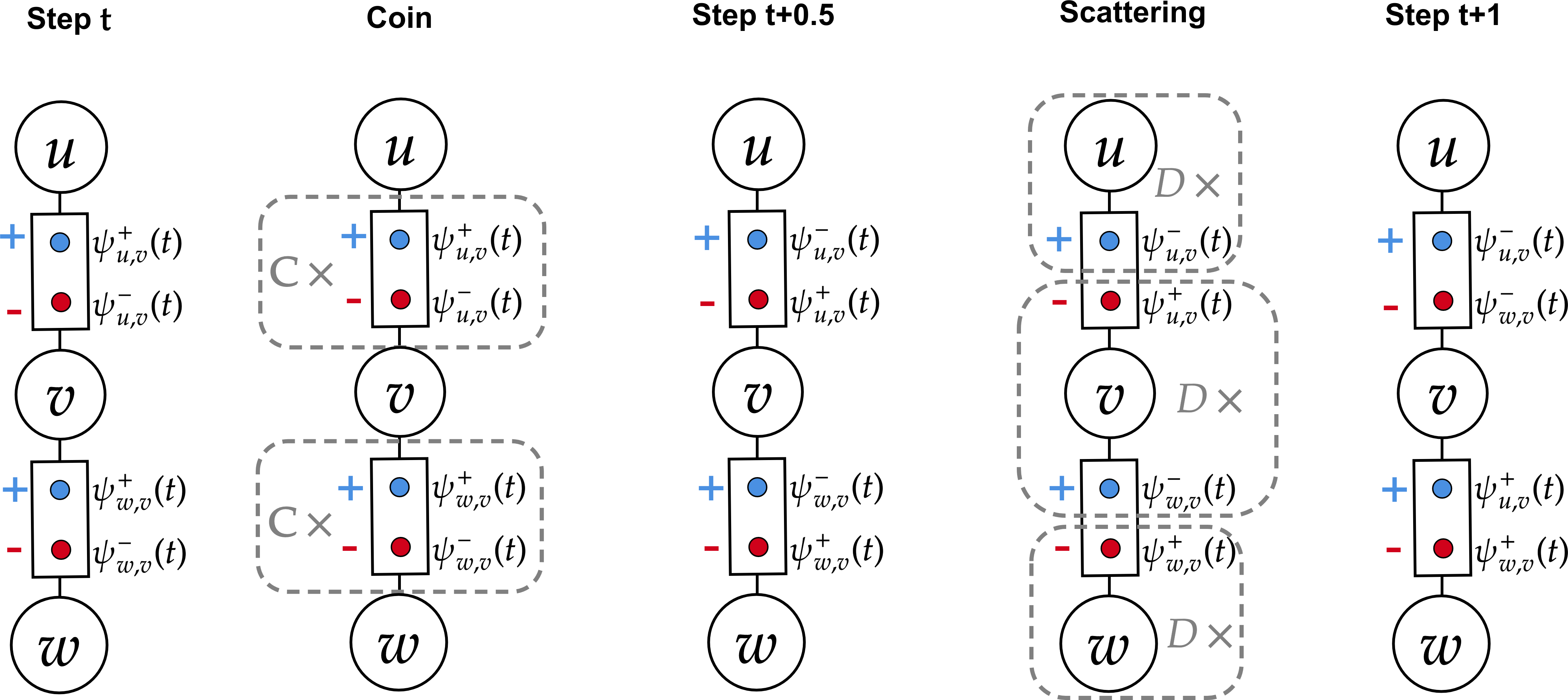}
	\caption{Example of a walk on a path of size 3.}
	\label{fig:ring}
\end{figure}

\paragraph{Polarity : why and how}
There are two reasons we need polarity. 
First, if the coin operator acts differently on $\ket{+}$ and $\ket{-}$, then different polarities lead to different dynamics. Polarities, in fact, determines how each edge's state scatter within the network, analogously to other spatial searching algorithms on regular lattice~\cite{bezerra2021quantum, arrighi2018dirac, roget2020grover}.
Moreover, polarity is used to divide the edges in two. So each node can only access one amplitude of each edge (depending on the polarity of the edges). This makes things much simpler for a distributed implementation, since we do not have to consider the case of two simultaneous operations on the same amplitude during scattering. 
One way to initialise the polarity on the graph is to create a coloring of $G$ at the initial time step. Afterward, edges have their $+$ pole in the direction of the node with the higher color. This design is especially convenient for bipartite graphs where every node sees only $+$ polarities or only $-$ polarities.

\subsection{Searching}
The above dynamical system formally describes a quantum walk on the graph's edges. If the quantum state is measured, we obtain one of the edges of the graph. This scheme is thus well suited to search edges. Thus we may now introduce an oracle marking the desired edge $(a,b) \in E$. Analogously with the standard spatial search, the oracle is defined as follows~:
$$
\mathcal O_{f} =\left(\sum_{f((u,v))=1}\ket{(u,v)}\bra{(u,v)}\right) \otimes R + \left(\sum_{f((u,v))=0}\ket{(u,v)}\bra{(u,v)}\right)\otimes I_2, 
$$
where $f$ is the classical oracle equals to $1$ if and only if the edge is marked, and $0$ otherwise. The operator $R$ is a special coin operator which is applied only to the marked edge. Without lack of generality, in the following we set $C=X$ and $R=-X$. The algorithm proposed here is the following~: 
\begin{algorithm}
	\caption{Search a marked edge}\label{algo:search}
	\begin{algorithmic}[1]
		\Require $G = (V,E)$ undirected, connected graph
		\Require $f$ the classical oracle
		\Require $T\in \mathbb{N}$ the hitting time
		\Function{Search}{$G,f,T$}
			\State $\displaystyle \ket{\Phi} \gets \frac{1}{\sqrt{2|E|}}\sum_{(u,v)\in E} \ket{(u,v)}\ket{+} + \ket{(u,v)}\ket{-}$ \Comment{Diagonal initial state.}
			\For{$0\leq i < T$}
				\State $\displaystyle \ket{\Phi} \gets S \times (I\otimes C) \times \mathcal O_{f} \times \ket{\Phi}$ \Comment{One step of the walk}
			\EndFor
			\State $(u,v,\pm) \gets$ \Call{Measure}{$\ket{\Phi}$} \Comment{Measure the quantum state.}
			\State \Return $(u,v)$
		\EndFunction
	\end{algorithmic}
\end{algorithm}
The initial state is initialized to be diagonal on the basis states and the additional oracle operation is added to the former QW-dynamic. There are two parameters that characterize the above algorithm~: the probability of success $P$ of returning the marked edge and the hitting time $T$. The number of oracle calls of this algorithm is $O(T)$. In practice, we want to choose $T$ such that $P$ is maximal.
The previous algorithm can easily be made into a Las Vegas algorithm whose answer is always correct but the running time random. Algorithm \ref{algo:guaranteed_search} shows such transformation. Interestingly, the expected number of times Algorithm \ref{algo:guaranteed_search} calls Algorithm \ref{algo:search} is $O(1/P)$ where $P$ is the probability of success of Algorithm \ref{algo:search}. The expected complexity of Algorithm \ref{algo:guaranteed_search} is $O\left(T/P\right).$
\begin{algorithm}
	\caption{Search a marked edge with guaranteed success}\label{algo:guaranteed_search}
	\begin{algorithmic}[1]
		\Require $G = (V,E)$ undirected, connected graph
		\Require $f$, the classical oracle
		\Require $T\in \mathbb{N}$ the hitting time
		\Function{GuaranteedSearch}{$G,f,T$}
		\State $(u,v) \gets (\texttt{nil},\texttt{nil})$ \Comment{Initial value}
		\While{$f((u,v)) \neq 1$} \Comment{Until we find the marked edge ...}
			\State $(u,v) \gets$ \Call{Search}{$G,f,T$} \Comment{... search again.}
		\EndWhile
		\State \Return $(u,v)$
		\EndFunction
	\end{algorithmic}
\end{algorithm}
The complexity of our algorithm is limited by the optimal complexity of $O(\sqrt{K})$ for searching problem where $K$ is the total number of elements. This is the complexity of the Grover algorithm which has been shown to be optimal \cite{zalka1999grover}. However Grover's algorithm assume full connectivity between elements which is not our case. The quantum walk presented here is on the edge of $G$, which means that the optimal complexity is $O\left(\sqrt{|E|}\right)$.

\subsection{Example : Searching an edge in the star graph}
In this example we consider the star graph with $M$ edges and $M+1$ nodes; a graph with $M+1$ nodes where every node (other than node $u_0$) is connected to node $u_0$. Our searching algorithm performs well on this graph, as shown by Theorem \ref{th:star}. In this section, we show the proof of Theorem \ref{th:star} which is mainly spectral analysis on a simplified dynamic.

\begin{theorem}\label{th:star}
	Let $G$ be the star graph with $M$ edges. Algorithm \ref{algo:search} has an optimal hitting time $T=O\left(\sqrt{M}\right)$ and a probability of success $O(1)$ for $G$. Algorithm \ref{algo:guaranteed_search} has expected complexity $O\left(\sqrt{M}\right)$.
\end{theorem}
\begin{proof}
	We consider the star graph $G = (V,E)$ of size $M+1$ such that
	$
	V = \{u_0,\ldots,u_M\} \text{ and } E = \{(u_0,u_i)\mid 1\leq i \leq M\}.
	$
	We assume without a loss of generality that the marked edge is $(u_0,u_1)$ and that the polarity is $\forall i>1, \; \sigma(u_0,u_i) = +$. At any time step $t$, the state of the walk reads $\ket{\Phi_t} = \sum_{i=1}^M \psi_{u_0,u_i}^+(t)\ket{(u_0,u_i)}\ket{+} + \psi_{u_0,u_i}^-(t)\ket{(u_0,u_i)}\ket{-}.$	
	We first show that $\forall t\in \mathbb{N}, \; \forall i > 1, \; \psi_{u_0,u_i}^+(t) = \alpha^+ \text{ and } \psi_{u_0,u_i}^-(t) = \alpha^-.$ This greatly simplifies the way we describe the walk dynamic. Afterward we shall provide simple spectral analysis to extract the optimal hitting time $T$ and probability of success $P$. Next, we prove 
	the property $(Q_t): \quad \forall i > 1, \; \psi_{u_0,u_i}^+(t) = \alpha_t^+ \text{ and } \psi_{u_0,u_i}^-(t) = \alpha_t^-$ for all $t\in \mathbb{N}$.
	
	The initial state $\ket{\Psi_0}$ is diagonal on the basis states~:
	$$ \ket{\Phi} \gets \frac{1}{\sqrt{2|E|}}\sum_{(u,v)\in E} \ket{(u,v)}\ket{+} + \ket{(u,v)}\ket{-}.$$
	All $\left(\psi_{u_0,u_i}^\pm(0)\right)_{i\geq 1}$ are equals so the property $Q_0$ is satisfied with $\alpha_t^+ = \alpha_t^- = \frac{1}{\sqrt{2M}}$. Now, we assume that $Q_t$ is true and use the walk dynamic to show that $Q_{t+1}$ is also true. The state $\ket{\Psi_{t+1}}$ is described in Table \ref{tab:star} and show that $Q_{t+1}$ is true.
	
\begin{table}[h]
	\caption{Detailed dynamic of the quantum walk for star graphs.}
	\label{tab:star}
	\begin{tabular}{l|cccl}
		& Step $t$ & After oracle & After coin & After Scattering \\
		\hline
		& & & & \\
		$\alpha^+$ & $\alpha^+_t$ & $\alpha^+_t$ & $\alpha^-_t$ & $\displaystyle\alpha^+_{t+1} = \frac{(M-2)\alpha_t^- - 2\psi_{u_0,u_1}^+}{M}(t)$ \\
		& & & & \\
		$\alpha^-$ & $\alpha^-_t$ & $\alpha^-_t$ & $\alpha^+_t$ & $\alpha^-_{t+1}=\alpha^+_t$ \\
		& & & & \\
		$\psi_{u_0,u_1}^+$ & $\psi_{u_0,u_1}^+(t)$ & $-\psi_{u_0,u_1}^-(t)$ & $-\psi_{u_0,u_1}^+(t)$ & $\displaystyle\psi_{u_0,u_1}^+(t+1) = \frac{(2M-2)\alpha_t^- +(M-2)\psi_{u_0,u_1}^+}{M}(t)$ \\
		& & & & \\
		$\psi_{u_0,u_1}^-$ & $\psi_{u_0,u_1}^-(t)$ & $-\psi_{u_0,u_1}^+(t)$ & $-\psi_{u_0,u_1}^-(t)$ & $\psi_{u_0,u_1}^-(t+1)=-\psi_{u_0,u_1}^-(t)$ \\
	\end{tabular}
\end{table}

Using the recurrence in Table \ref{tab:star}, we can put the dynamic of the walk into a matrix form.
$$
\psi_{u_0,u_1}^-(t) = \frac{(-1)}{\sqrt{2M}} \qquad \text{and} \qquad X_{t+1} = AX_t,
$$
where
$$
X_t = \mat{\alpha^+_t\\\alpha_t^-\\\psi_{u_0,u_1}^+(t)\\}\qquad \text{and} \qquad A = \mat{
	0 & \frac{M-2}{M} & \frac{-2}{M}\\
	1 & 0 & 0\\
	0 & 2\frac{M-1}{M} & \frac{M-2}{M}\\
}.
$$

After diagonalizing $A$, the eigenvalues are $\{-1,e^{i\lambda},e^{-i\lambda}\}$, where $e^{i\lambda} = \frac{M-1+i\sqrt{2M-1}}{M}$. This leads to
$
\psi_{u_0,u_1}^+(t) \sim \frac{i}{2}\left(e^{-i\lambda t} - e^{i\lambda t}\right) \sim \sin(\lambda t)$, which allows us to deduce the probability $p_t$ of hitting the marked edge
$
p_t \sim \sin^2(\lambda t),
$
and then the optimal hitting time by solving $p_T = 1$:
$$
T = \frac{\pi}{2} \frac{1}{\lambda} \sim \frac{\pi}{2} \sqrt{\frac{M}{2}} \sim \frac{\pi}{2\sqrt{2}} \sqrt{M} = O(\sqrt{M}).
$$
Finally we have a probability of success $P\sim 1$ for the optimal hitting time $\displaystyle T \sim \frac{\pi}{2\sqrt{2}} \sqrt{M} = O(\sqrt{M})$.

\end{proof}

\subsection{Searching nodes}
The searching quantum walk presented in the previous section can search one marked edge in a graph. In order to search a node instead, we need to transform the graph we walk on. We call this transformation \textsc{starify}.

\begin{definition}{\textsc{Starify}}\\
	Let us consider an undirected graph $G = (V,E)$. The transformation \textsc{starify} on $G$ returns a graph $\tilde G$ for which
	\begin{itemize}
		\item every node and edge in $G$ is in $\tilde G$ (we call them real nodes and real edges).
		\item for every node $u\in V$, we add a node $\tilde u$ (we call these new nodes virtual nodes).
		\item for every node $u\in V$, we add the edge $(u,\tilde u)$ (we call these new edges virtual edges).
	\end{itemize}
	We call the resulting graph $\tilde G$ the starified graph of $G$.
\end{definition}

\paragraph{Searching nodes} Starifying a graph $G$ allows us to search a marked node $u$ instead of an edge.We can then use the previous searching walk to search the virtual edge and then deduce without ambiguity the marked node $u$ of the initial graph $G$. This procedure implies that we have to increase the size of the graph (number of edges and nodes). In particular, increasing the number of edges is significant since we increase the dimension of the walker (this can be significant for sparse graphs) and one must be careful of this when computing the complexity. A reassuring result is that searching a node on the complete graph (the strongest possible connectivity) has optimal complexity. As stated in Theorem \ref{th:complete}, the complexity is $O(\sqrt{M})$, which is optimal when searching over the edges. Compared to a classical algorithm in $O(M)$ (a depth first search for instance), we do have a quadratic speedup.

\begin{theorem}\label{th:complete}
	When using the quantum algorithm to search one marked node in the starified graph $\tilde G$ of the complete graph $G$ of size $N$, the probability of success is $P \sim 1$ and the hitting time $T \sim \frac{\pi}{4} N$.
\end{theorem} 
\begin{proof}
	The proof is very similar to the one of Theorem \ref{th:star}. We show that several edges have the same states to reduce the dynamic to a simple recursive equation. We then solve it numerically and derive asymptotic values for the probability of success. A detailed proof shall be provided in the supplementary materials.
\end{proof}

\section{Distributed Implementation}\label{sec:impl}
In this section we move to the multi-particle states dynamics, allowing a distributed implementation of the above searching protocol. We assume that both nodes and edges have quantum and classical registers. They can transmit classical bit and apply controlled operations over the local neighborhood. The nodes need to know their degree but do not require a unique label. We first explain how the quantum registers are dispatched along the graph, then we introduce a distributed protocol to reproduce the dynamic of the walk. In this section we assume that we want to implement the walk on a graph $G=(V,E)$ with polarity $\sigma$.

\subsection{Positioning the qubits}
The implementation proposed uses $W$ states\cite{cruz2019efficient}; i.e. a state of the form $(\ket{001}+\ket{010}+\ket{100})/\sqrt{3}$. Basically the superposition of all the unary digits. We use two qubits per edge, one for the $+$ state and one for the $-$ state. For every node $u\in V$, we use $\ln(\text{deg}(u))+1$ qubits to apply the diffusion across all the neighborhood.

\paragraph{Edges register}
The edge register represents the position of the walker and is the one measured at the end of the algorithm. It consists of two qubits per edge, corresponding respectively to the + and - polarity amplitudes.
To represent a unary state in this register, we need to introduce the notation  $\delta_k^n = \underbrace{0\ldots0}_{k-1 \text{ times}}1\underbrace{0\ldots0}_{n-k \text{ times}}$ for the string of $n$ digit with $0$ everywhere except at position $k$ where it has $1$. For instance, $\delta_2^4 = 0100$.
For the sake of simplifying the notations, let us arbitrary enumerate the edges of $E$ such that we have $E = \{e_1,\ldots,e_{|E|}\}$. The full state of the walk is a linear combination of all the $\left( \ket{\delta_k^{2|E|}}\right) _{k \in 2|E|}$, where $\ket{\delta_{2k}^{2|E|}}$ corresponds to amplitude $\ket{e_k}\ket{+}$ of the mathematical model, and $\ket{\delta_{2k+1}^{2|E|}}$ corresponds to amplitude $\ket{e_k}\ket{-}$. It is important that the global state of the register remains a superposition of $\left( \ket{\delta_k^{2|E|}}\right) _{k \in 2|E|}$, as this allows us to obtain a valid solution to the search problem during measurement. In fact, no matter which state $\delta_k^{2|E|}$ is measured, all the edges are measuring state 0 except one, which is measuring state 1.

\paragraph{Nodes register}
The node register consists of $\log \deg +1$ qubits per node. This register is made up exclusively of auxiliary qubits used for the scattering operation. During the scattering operation, the amplitudes of the edges around a given node $u$ are moved to $u$'s auxiliary qubits. The $u$'s qubits then holds the local amplitudes in binary format, applies the scattering operation and finally transfers these amplitudes back to the local edges. After every step of the quantum walk, the nodes register is in state $\ket{0\ldots0}$ without any entanglement with the edges register.

\subsection{Distributed protocol}
The various quantum registers, their sizes, and how they are positionned on the network has been made clear in the previous section. We know give distributed schemes for every operation of the mathematical model: oracle, coin and scattering.

\paragraph{Oracle}
In the mathematical model, the oracle applies $R$ on the marked edge only. A definition on the basis of this operator when $e_1$ is marked would be 
$$
\mat{\ket{e_1}\ket{+}\\\ket{e_1}\ket{-}\\} \overset{oracle}\longmapsto R\mat{\ket{e_1}\ket{+}\\\ket{e_1}\ket{-}\\}.
$$ and 
$$\qquad \qquad \ket{e_k}\ket{s} \overset{oracle}\longmapsto \ket{e_k}\ket{s} $$
for all $2\leq k\leq |E|$,  $\forall s = \pm$.\\
Translated to the distributed registers it means that 
$$
\mat{\ket{\delta_{2k}^{2|E|}}\\\ket{\delta_{2k+1}^{2|E|}}\\} \overset{oracle}\longmapsto R\mat{\ket{\delta_{2k}^{2|E|}}\\\ket{\delta_{2k+1}^{2|E|}}\\}, \qquad \text{and identity everywhere else}.
$$
We recall that in our searching algorithm, $R=-X$. This operation can be achieved by two $Z$ gates and a swap applied only to the marked edge.

\paragraph{Coin}
The coin operation is very similar to the oracle. We want to apply $C=X$ to all edge states. In the mathematical model, it can be written as
$$
\forall k, \; \mat{\ket{e_k}\ket{+}\\\ket{e_k}\ket{-}\\} \overset{oracle}\longmapsto C\mat{\ket{e_k}\ket{+}\\\ket{e_k}\ket{-}\\} .
$$
This can be realized by a swap applied on all edges.

\paragraph{Scattering}
The scattering operation consists in applying a diffusion operator $D$ locally around the nodes. Since scattering around each node is affecting distinct qubits, we only need to design the circuit for one node $u$ of degree $d$. We note $a_1,\ldots,a_d$ the local edges connected to $u$. For this section we call $\eta_k$ the qubits of edge $a_k$ accessible to $u$ following the polarity. Note that, while the local edges (and their qubits) are indexed, this is completely arbitrary and is not changing the final result of the scattering. This nice property is due to the shift invariance of diffusion operator $D$. The $\log d$ first quibts of $u$ are storing the amplitudes of the local edges in a binary format, while the last one is marking the states we need to act on. We note $k^{(2)}$ the binary representation of $k$. We define the operator $\text{Tr}_k$ between qubits $\eta_k$ and $u$ register such that
$$
	\ket{0}\ket{0^{(2)}}\ket{0} \overset{\text{Tr}_k}\longmapsto \ket{0}\ket{0^{(2)}}\ket{0} \qquad \text{and}\qquad 
	\ket{1}\ket{0^{(2)}}\ket{0} \overset{\text{Tr}_k}\longmapsto \ket{0}\ket{(k-1)^{(2)}}\ket{1}.\\
$$
This operator $\text{Tr}_k$ admits a circuit depending of $k$ composed of X, CNOT and multi controlled Tofolli gates. Algorithm \ref{algo:Trk} shows a distributed scheme for realizing this circuit. It is running on the node and uses two communication methods: \texttt{RequestCnot}(edge, target) which applies a Not on target controlled by the qubit in edge accessible according to polarity; and \texttt{ApplyMCT}(edge) which applies a Not on the edge's qubit controlled by the full node register. In the worst case, this algorithm require $\log d +1$ calls to \texttt{RequestCnot} and one call to \texttt{ApplyMCT}.
We define Tr by the successive applications of $\text{Tr}_k$ between edge $a_k$ and node $u$ for all $k$. It is worth to note that, while one has to apply every $\text{Tr}_k$ successively, the application order does not matter. On the direct basis, the complete scattering operation we want to achieve can be written as 
$$
\underbrace{\ket{\delta_k^d}}_{\text{edges}}\ket{0^{(2)}}\ket{0} \overset{\text{Tr}}\longmapsto \ket{0}\ket{k^{(2)}}\ket{1} \overset{D}\longmapsto \ket{0}D\ket{k^{(2)}}\ket{1} \overset{\text{Tr}^{-1}}\longmapsto \underbrace{\left( \sum_{j=1}^dD_{k,j}\ket{\delta_j^d}\right)}_{\text{walker scattered}} \otimes\ket{0^{(2)}}\ket{0}.
$$ 
The last qubit allows us to apply $D$ only when the walker was actually on edge $k$. Following this global idea, Algorithm \ref{algo:Tr} provides a distributed implementation for the scattering operator. Every node needs a total of $O(d\log d)$ controlled operations (\texttt{RequestCnot} or \texttt{ApplyMCT}) with the local edges.

Finally, Figure \ref{fig:circuit} shows an example of this distributed design of the path graph of five nodes. The path graph has the particularity of having the same topology as a circuit (a qubit being connected to the preceding and following qubit). Notice that such circuit coincied with a partitioned QCA, each operation is local and traslationally invariant. 

\begin{algorithm}[h]
	\caption{Distributed scheme for $\text{Tr}_k$ on node $u$ of degree $d$}
	\label{algo:Trk}
	\begin{algorithmic}[1]
		\Require $e$ an edge connected to $u$
		\Require $1\leq k \leq d$
		\Function{$\text{Tr}_k$}{$e,k$}
		\State $r \gets \lceil\log d\rceil$ \Comment{The size of $u$'s register is $r+1$}
		\For{$0\leq i < r \mid \left((k-1)^{(2)} \right)_i = 1$}
		\State \Call{RequestCnot}{$e, q_i$} \Comment{CNOT on the $i^{\text{th}}$qubit of $u$}
		\EndFor
		
		\State \Call{RequestCnot}{$e, q_r$}
		
		\State Apply $X$ gate on all qubits $q_i$ of $u$ such that $\left((k-1)^{(2)} \right)_i = 0$.
		
		\State \Call{ApplyMCT}{e}
		
		\State Apply $X$ gate on all qubits $q_i$ of $u$ such that $\left((k-1)^{(2)} \right)_i = 0$.
		
		\EndFunction
	\end{algorithmic}
\end{algorithm}

\begin{algorithm}[h]
	\caption{Distributed scheme for Tr on node $u$ of degree $d$}
	\label{algo:Tr}
	\begin{algorithmic}[1]
		\Require $\mathcal{N}$ the set of all edges connected to $u$
		\Function{$\text{Tr}$}{$u$}
		\State $\{a_1,\ldots,a_d\} = \mathcal{N}$ \Comment{Random or arbitrary enumeration of the edges}
		\For{$1\leq k\leq d$}
		\State \Call{$\text{Tr}_k$}{$a_k,k$}
		\EndFor
		\EndFunction
	\end{algorithmic}
\end{algorithm}

\begin{figure}[h]
	\centering
	\includegraphics[width=\textwidth]{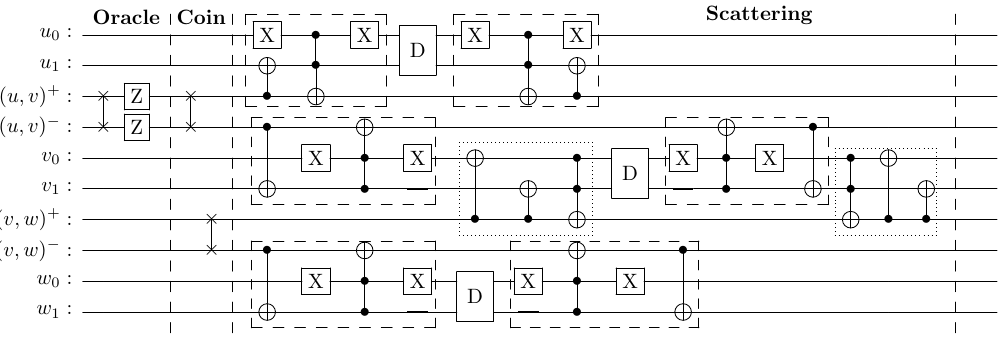}
	\caption{Circuit of one step of the quantum walk for the path graph $u-v-w$. Dashed lines signal $\text{Tr}_1$ circuit and its inverse while dotted lines $\text{Tr}_2$. The circuit applies successively the oracle on $(u,v)$, the coin, Tr, $D$, $\text{Tr}^{-1}$.}
	\label{fig:circuit}
\end{figure}

\section*{Discussion}
The distributed algorithm described in Section \ref{sec:impl} is based on the quantum walker dynamics introduced earlier for the single particle sector. The most important feature of the above model is that it does not require a leader or unique identifiers. Once the measure is done, all edges measure $0$ except one that measure $1$, giving us a valid solution; i.e. a solution where every edge agrees on who is marked. The edge who measured $1$ can then send the answer to some hypervisor or broadcast it across the network, depending of the context. Notice that this distributed scheme requires that the network be in a $W$ state, which is very common and widely studied in the distributed algorithm community~\cite{cruz2019efficient,tani2012exact}. To study the performances of this algorithm, one should use the mathematical model from Section \ref{sec:maths}. In general, it is hard to study analytically this model, except in some rare cases, where a lot of edges have the same state during the whole execution. Indeed, the topology of the dynamics, which is different from that of the graph, makes its analysis mathematically difficult. In fact, we have to consider a topology based on edges, where two edges are neighbors if they share a node. As such, regular lattices like a grid do not actually have the topology of a grid. However, one could also study this algorithm numerically; preliminary non published work on this aspect have already shown promising results for popular lattices and random graphs. To conclude, it is also important to note that the distributed model recover a QCA-like dynamics on an arbitrary graph. Indeed, in the case of a linear graph, it coincides with the standard partitioned QCA. Overall each operation are node- and edge-independent, so they are translationally invariant. And all operations are local. A byproduct of these properties is a promising distributed architecture which naturally gives a degree of fault tolerance because of local interactions, limiting error propagation.

\paragraph{Aknowledgement}
This work is supported by the PEPR EPiQ ANR-22-PETQ-0007, by the ANR JCJC DisQC ANR-22-CE47-0002-01.

\bibliographystyle{splncs04}
\bibliography{sn-bibliography}

\end{document}